\documentclass[11pt]{article}
\usepackage{amsmath,amsthm,amssymb,bm,xspace,booktabs}

\usepackage[margin=2.5cm]{geometry}
\usepackage[colorlinks=true]{hyperref}
\usepackage{url}

\usepackage{color}
\usepackage{graphicx}

\usepackage[T1]{fontenc}
\newtheorem{theorem}{Theorem}[section]
\newtheorem{lemma}[theorem]{Lemma}

\newtheorem{claim}{Claim}

\theoremstyle{plain}

\newenvironment{claimproof}{\begin{proof}}{\end{proof}}
\usepackage{lineno}
\nolinenumbers

\newcommand{\Zzero}{\mathbb{Z}_{\ge 0}}

\newcounter{one}
\newcommand{\one}{{\rm \roman{one}}}
\newcounter{two}
\newcommand{\two}{{\rm \roman{two}}}

\newcounter{three}
\newcounter{four}
\newcounter{five}
\newcounter{six}
\newcommand{\NP}{{$\mathsf{NP}$}\xspace}
\newcommand{\PSPACE}{{$\mathsf{PSPACE}$}\xspace}
\newcommand{\coNP}{{$\mathsf{coNP}$}\xspace}

\newcommand{\APX}{{$\mathsf{APX}$}\xspace}

\newenvironment{listing}[1]{%
        \begin{list}{*}{%
                 \settowidth{\labelwidth}{#1}%
                 \setlength{\leftmargin}{\labelwidth}%
                  \advance \leftmargin by 12pt
                   \setlength{\itemsep}{0pt}%
                   \setlength{\parsep}{0pt}%
                   \setlength{\topsep}{0pt}%
                   \setlength{\parskip}{0pt}%
}%
}{%
\end{list}}

\newcommand{\R}{\mathbb{R}}

\newcommand{\symmdiff}{\mathbin{\triangle}}

\newcommand{\opt}{\mathrm{OPT}}
\newcommand{\prob}{\textsc{Shortest Perfect Matching Reconfiguration}\xspace}

\newcommand{\TreeDiamDec}{\textsc{Min-Sum Diameter Decomposition}\xspace}
\newcommand{\ham}{\textsc{Hamiltonian Cycle Problem}\xspace}
\newcommand{\diham}{\textsc{Directed Hamiltonian Cycle Problem}\xspace}

\title{Shortest Reconfiguration of Perfect Matchings\\ via Alternating Cycles}

\date{}
\author{
 Takehiro Ito\thanks{Partially supported by JST CREST Grant Number JPMJCR1402, and JSPS KAKENHI Grant Numbers JP18H04091 and JP19K11814, Japan.}\\ Tohoku University, Japan\\  \texttt{takehiro@ecei.tohoku.ac.jp}
  \and
  Naonori Kakimura\thanks{Supported by JSPS KAKENHI Grant Numbers JP17K00028 and JP18H05291.}\\ Keio University, Japan\\  \texttt{kakimura@math.keio.ac.jp}
  \and
  Naoyuki Kamiyama\thanks{Partially supported by JST PRESTO Grant Number JPMJPR1753, Japan.}\\ Kyushu University, and JST, PRESTO, Japan\\ \texttt{kamiyama@imi.kyushu-u.ac.jp}
  \and
  Yusuke Kobayashi\thanks{Partly supported by JSPS KAKENHI Grant Numbers JP16K16010, JP17K19960, and JP18H05291, Japan.}\\ Kyoto University, Japan\\ \texttt{yusuke@kurims.kyoto-u.ac.jp}
  \and
  Yoshio Okamoto\thanks{Partially supported by JSPS KAKENHI Grant Number 15K00009 and JST CREST Grant Number JPMJCR1402, and Kayamori Foundation of Informational Science Advancement.}\\ University of Electro-Communications, and\\ RIKEN Center for Advanced Intelligence Project, Japan\\ \texttt{okamotoy@uec.ac.jp}
}

\begin{document}
 \maketitle
\begin{abstract}
  Motivated by adjacency in perfect matching polytopes,
  we study the shortest reconfiguration problem of perfect matchings via
  alternating cycles.
  Namely, we want to find a shortest sequence of perfect matchings which transforms one
  given perfect matching to another given perfect matching such that
  the symmetric difference of each pair of
  consecutive perfect matchings is a single cycle.
  The problem is equivalent to the combinatorial shortest path problem
  in perfect matching polytopes.
  We prove that the problem is {\NP}-hard even when a given graph is
  planar or bipartite, but it can be solved in polynomial time when
  the graph is outerplanar.
\end{abstract}


\section{Introduction}

\emph{Combinatorial reconfiguration} is a fundamental research subject
that sheds light on solution spaces of combinatorial (search) problems, and
connects various concepts such as optimization, counting, enumeration, 
and sampling.
In its general form, combinatorial reconfiguration is concerned with
properties of the configuration space of a combinatorial problem.
The configuration space of a combinatorial problem is often represented
as a graph, but its size is usually exponential in the instance size.
Thus, algorithmic problems on combinatorial reconfiguration are not
trivial, and require novel tools for resolution.
For recent surveys, see~\cite{DBLP:journals/algorithms/Nishimura18,DBLP:books/cu/p/Heuvel13}.

Two basic questions have been encountered in the study of combinatorial reconfiguration.
The first question asks the existence of a path between two given solutions in
the configuration space, namely the \emph{reachability} of the two solutions.
The second question asks the shortest length of a path between two given 
solutions, if it exists.
The second question is usually referred to as a 
\emph{shortest reconfiguration problem}.

\begin{figure}[tb]
  \centering
		\includegraphics[width=\linewidth]{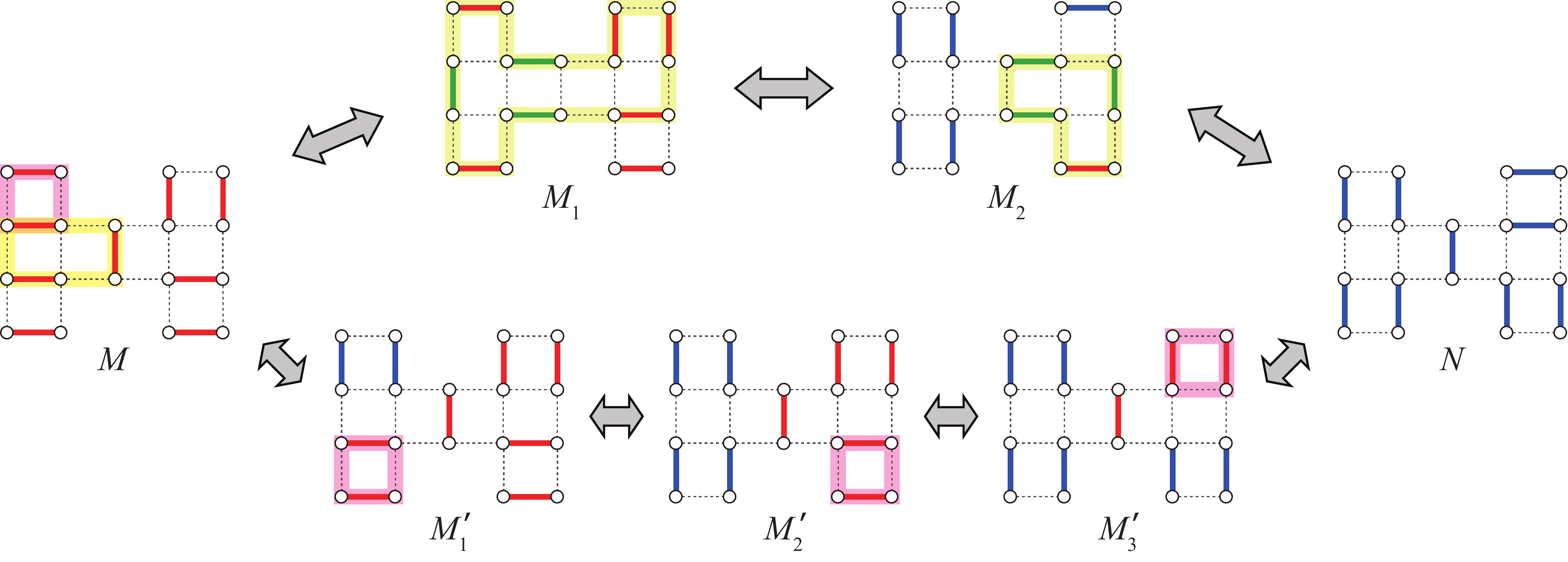}
	\caption{Two sequences of perfect matchings between $M$ and $N$ under the alternating cycle model. 
		The sequence $\langle M, M_1, M_2, N \rangle$ following the yellow alternating cycles is shortest even though it touches the edge in $M \cap N$ twice. 
		On the other hand, $\langle M, M_1^\prime, M_2^\prime, M_3^\prime, N \rangle$ following the pink alternating cycles is not shortest although it touches only the edges in $M \symmdiff N$.}
	\label{fig:example}
\end{figure}

In this paper, we focus on reconfiguration problems of matchings, namely sets of
independent edges.
There are several ways of defining the configuration space for matchings, and some of them have already been studied in the literature~\cite{DBLP:journals/tcs/ItoDHPSUU11,DBLP:journals/tcs/KaminskiMM12,DBLP:conf/sofsem/GuptaKM19,DBLP:journals/corr/abs-1812-05419,PMRarXiv}.
We will explain them in Section~\ref{subsec:closelyrelatedwork}.

We study yet another configuration space for matchings, 
which we call the \emph{alternating path/cycle} model.
The model is motivated by adjacency in matching polytopes, which we will see
soon.
In the model,
we are given an undirected and unweighted graph $G$,
and also an integer $k \geq 0$.
The vertex set of the configuration space consists of 
the matchings in $G$ of size at least $k$.
Two matchings $M$ and $N$ in $G$ are adjacent in the configuration space if and only if their symmetric
difference $M \symmdiff N := (M\cup N) \setminus (M\cap N)$ is a single path or cycle.
In particular, we are interested in the case where $k=|V(G)|/2$, namely
the reconfiguration of \emph{perfect matchings}.
In that case, the model is simplified to the \emph{alternating cycle} model
since $M \symmdiff N$ cannot have a path.
See \figurename~\ref{fig:example} as an example.

The reachability of two perfect matchings is trivial under the alternating cycle model: the answer is always yes.
This is because the symmetric difference of two perfect matchings always
consists of vertex-disjoint cycles.
Therefore, we focus on the shortest perfect matching reconfiguration 
under the alternating cycle model.

\subsection{Related Work} \label{subsec:closelyrelatedwork}
\paragraph*{Other Configuration Spaces for Matchings}
As mentioned, reconfiguration problems of matchings have already been studied under different models~\cite{DBLP:journals/tcs/ItoDHPSUU11,DBLP:journals/tcs/KaminskiMM12,DBLP:conf/sofsem/GuptaKM19,DBLP:journals/corr/abs-1812-05419,PMRarXiv}.
These models chose more elementary changes as the adjacency on the configuration space.  
Then, the situation changes drastically under such models: 
even the reachability of two matchings is not guaranteed. 

Matching reconfiguration was initiated by the work of
Ito et al.~\cite{DBLP:journals/tcs/ItoDHPSUU11}.
They proposed the \emph{token addition/removal} model of reconfiguration,
in which 
we are also given an integer $k \geq 0$, and
the vertex set of the configuration space consists of 
the matchings of size at least $k$.\footnote{Precisely, their model is defined in a slightly different way, but it is essentially the same as this definition.}
Two matchings $M$ and $N$ are adjacent if and only if they differ in 
only one edge.
Ito et al.~\cite{DBLP:journals/tcs/ItoDHPSUU11} proved that the
reachability of two given matchings can be checked in polynomial time.

Another model of reconfiguration is \emph{token jumping}, introduced by Kami\'nski et al.~\cite{DBLP:journals/tcs/KaminskiMM12}.
In the token jumping model, 
we are also given an integer $k \geq 0$, and
the vertex set of the configuration space consists of 
the matchings of size exactly $k$.
Two matchings $M$ and $N$ are adjacent if and only if they differ in 
only two edges.
Kami\'nski et al.~\cite[Theorem~1]{DBLP:journals/tcs/KaminskiMM12} proved that the token jumping model is equivalent to the token addition/removal model when $|M| = |N|$. 
Thus, using the result by Ito et al.~\cite{DBLP:journals/tcs/ItoDHPSUU11},
the reachability can be checked in polynomial time also under the token jumping model~\cite[Corollary~2]{DBLP:journals/tcs/KaminskiMM12}.

On the other hand, the shortest matching reconfiguration is known to
be hard.
Gupta et al.~\cite{DBLP:conf/sofsem/GuptaKM19} and
Bousquet et al.~\cite{DBLP:journals/corr/abs-1812-05419}
independently
proved that
the problem is \NP-hard under the token jumping model.
Then, the problem is also \NP-hard under the
token addition/removal model,
because the shortest lengths are preserved under the two models~\cite[Theorem~1]{DBLP:journals/tcs/KaminskiMM12}.

Recently, Bonamy et al.~\cite{PMRarXiv} studied the reachability of two perfect matchings under a model close to ours, namely the alternating cycle model \emph{restricted to length four}. 
In the model, two perfect matchings $M$ and $N$ are adjacent if and only if their symmetric
difference $M \symmdiff N$ is a cycle of length four.
Then, the answer to the reachability is not always yes, and Bonamy et al.~\cite{PMRarXiv} proved that the reachability problem is \PSPACE-complete under this restricted model.

\paragraph*{Relation to Matching Polytopes}
Our alternating cycle model (without any restriction of cycle length) for the perfect matching reconfiguration
is natural when we see the connection with the simplex methods for
linear optimization, or combinatorial shortest paths of the graphs of
convex polytopes.

In the combinatorial shortest path problem of a convex polytope,
we are given a convex polytope $P$, explicitly or implicitly, and
two vertices $v,w$ of $P$.
Then, we want to find a shortest sequence $u_0,u_1,\dots,u_t$ of vertices
of $P$ such that $u_0=v, u_t=w$ and $\overline{u_{i}u_{i+1}}$ forms an edge of $P$
for every $i=0,1,\dots,t-1$.
Often, we are only interested in the length of such a shortest sequence,
and we are also interested in the maximum shortest path length among 
all pairs of vertices, which is known as the combinatorial diameter of 
the polytope $P$.
The combinatorial diameter of a polytope has attracted much attention
in the optimization community from the motivation of better understanding
of simplex methods.
Simplex methods for linear optimization start at a vertex of the feasible
region, follow edges, and arrive at an optimal vertex.
Therefore, the combinatorial diameter dictates the best-case behavior of
such methods.
The famous Hirsch conjecture states that every $d$-dimensional convex 
polytope with $n$ facets has the combinatorial diameter at most $n-d$.
This has been disproved by Santos~\cite{santos12}, 
and the current best upper bound of $(n-d)^{\log_2 O(d/\log d)}$ for the 
combinatorial diameter was given by Sukegawa~\cite{sukegawa}.
On the other hand, for the 0/1-polytopes
 (i.e., polytopes in which the coordinates 
of all vertices belong to $\{0,1\}$), the Hirsch conjecture holds~\cite{DBLP:journals/mp/Naddef89}.

The shortest perfect matching reconfiguration under the alternating cycle model
can be seen as the
combinatorial shortest path problem of a perfect matching polytope.
The \emph{perfect matching polytope} of a graph $G$ is defined as follows.
The polytope lives in $\R^{E(G)}$, namely each coordinate corresponds to
an edge of $G$.
Each vertex $v$ of the polytope corresponds to a perfect matching $M$ of $G$ as
$v_e=1$ if $e \in M$ and $v_e=0$ if $e \not\in M$.
The polytope is defined as the convex hull of those vertices.
It is known that two vertices $u,v$ of the perfect matching polytope form
an edge if and only if the corresponding perfect matchings $M,N$ have the
property that $M \symmdiff N$ contains only one cycle~\cite{CHVATAL1975138}.
This means that the graph of the perfect matching polytope is exactly the
configuration space for perfect matchings under the
alternating cycle model.

\paragraph*{Further Related Work}

As mentioned before, the matching reconfiguration has been studied by
several authors~\cite{DBLP:journals/tcs/ItoDHPSUU11,DBLP:journals/tcs/KaminskiMM12,DBLP:conf/sofsem/GuptaKM19,DBLP:journals/corr/abs-1812-05419,PMRarXiv}.
Extension to $b$-matchings has been considered, too~\cite{DBLP:conf/mfcs/Muhlenthaler15,DBLP:journals/jco/ItoKKKO19}.

Shortest reconfiguration has attracted considerable attention.
Starting from an old work on the $15$-puzzle~\cite{DBLP:conf/aaai/RatnerW86},
we see the work on pancake sorting~\cite{DBLP:journals/jcss/BulteauFR15},
triangulations of point sets~\cite{DBLP:journals/comgeo/LubiwP15,DBLP:journals/comgeo/Pilz14} and
simple polygons~\cite{DBLP:journals/dcg/AichholzerMP15} under flip distances,
and also independent set reconfigurations~\cite{DBLP:conf/walcom/YamadaU16},
satisfiability reconfiguration~\cite{DBLP:journals/siamdm/MouawadNPR17},
coloring reconfiguration~\cite{DBLP:journals/algorithmica/0001KKPP16},
token swapping problems~\cite{DBLP:journals/tcs/YamanakaDIKKOSS15,DBLP:conf/esa/MiltzowNORTU16,DBLP:journals/tcs/YamanakaHKKOSUU18,DBLP:journals/algorithmica/BonnetMR18,DBLP:journals/jgaa/YamanakaDHKNOSS19,DBLP:journals/jgaa/KawaharaSY19}.
A tantalizing open problem is to determine the complexity of
computing the rotation distance of two rooted binary trees
(or equivalently the flip distance of two triangulations
of a convex polygon, or the combinatorial shortest path of an
associahedron).

The computational aspect of the combinatorial shortest path problem on convex
polytopes is not well investigated.
It is known that the combinatorial diameter is hard to determine~\cite{DBLP:journals/cc/FriezeT94} even for fractional matching polytopes \cite{DBLP:conf/focs/Sanita18}.
In the literature, we find many papers on the adjacency of convex polytopes
arising from combinatorial optimization problems~\cite{DBLP:journals/cor/GeistR92,matsuiMETR9412,DBLP:journals/dam/AlfakihM98,DBLP:journals/ejc/Fiorini03}.
Among others, Papadimitriou~\cite{DBLP:journals/mp/Papadimitriou78} 
proved that determining whether two given vertices are adjacent in a traveling salesman polytope is \coNP-complete.
This implies that computing the combinatorial shortest path between two vertices of a traveling salesman polytope is \NP-hard.
However, to the best of the authors' knowledge, all known combinatorial polytopes with such adjacency hardness stem from \NP-hard combinatorial optimization problems and the associated polytopes have exponentially many facets.
We also point out the work on a randomized algorithm to compute a combinatorial ``short'' path~\cite{DBLP:conf/icalp/BrunschR13}.

\subsection{Our Contribution}
To the best of the authors' knowledge, known results under different models do not have direct relations to our alternating cycle model, 
because their configuration spaces are different.
In this paper, we thus investigate the polynomial-time solvability of the shortest perfect matching reconfiguration under the alternating cycle model.  
The results of our paper are two-fold.
\begin{enumerate}
\item The shortest perfect matching reconfiguration under the alternating cycle model can be solved in polynomial time if the input graph is outerplanar.
\item The shortest perfect matching reconfiguration under the alternating cycle model is \NP-hard even when the input graph is planar or bipartite.
\end{enumerate}
Since outerplanar graphs form a natural and fundamental subclass of 
planar graphs, our results exhibit a tractability border among planar
graphs.

The hardness result for bipartite graphs implies that the computation
of a combinatorial shortest path in a convex polytope is \NP-hard even
when an inequality description is explicitly given.  This is because a
polynomial-size inequality description of the perfect matching
polytope can be explicitly written down from a given bipartite graph.

We point out that the hardness results have been independently obtained
by Aichholzer et al.~\cite{DBLP:journals/corr/abs-1902-06103}.
Indeed, they proved the hardness for planar bipartite graphs
(i.e., an input graph is planar \emph{and} bipartite).

\paragraph*{Technical Key Points}
Compared to recent algorithmic developments on reachability problems, only a few polynomial-time solvable cases are known for shortest reconfiguration problems.  
We now explain two technical key points, especially for algorithmic results on shortest reconfiguration problems. 

The first point is the symmetric difference of two given solutions.
Under several known models (not only for matchings) that employ elementary changes as the adjacency on the configuration space, the symmetric difference gives a (good) lower bound on the shortest reconfiguration. 
This is because any reconfiguration sequence (i.e., a path in the configuration space) between two given solutions must touch all elements in their symmetric difference at least once.  
For example, in \figurename~\ref{fig:example}, the symmetric difference of two perfect matchings $M$ and $N$ consists of $16$ edges and hence it gives the lower bound of $16/4 = 4$ under the alternating cycle model restricted to length $4$~\cite{PMRarXiv}. 
In such a case, if we can find a reconfiguration sequence touching only the elements in the symmetric difference (e.g., the sequence $\langle M, M_1^\prime, M_2^\prime, M_3^\prime, N \rangle$ in \figurename~\ref{fig:example}), then it is automatically the shortest under that model. 
However, this useful property does not hold under our alternating cycle model, because the length of an alternating cycle for reconfiguration is not fixed. 

The second point is the characterization of \emph{unhappy moves} that touch elements contained commonly in two given solutions. 
For example, the shortest reconfiguration sequence $\langle M, M_1, M_2, N \rangle$ in \figurename~\ref{fig:example} has an unhappy move, since it touches the edge in $M \cap N$ twice. 
In general, analyzing a shortest reconfiguration becomes much more difficult if such unhappy moves are required.
A well-known example is the (generalized) $15$-puzzle~\cite{DBLP:conf/aaai/RatnerW86} in which the reachability can be determined in polynomial time, while the shortest reconfiguration is \NP-hard. 
As illustrated in \figurename~\ref{fig:example}, the shortest perfect matching reconfiguration requires unhappy moves even for outerplanar graphs, and hence we need to characterize the unhappy moves to develop a polynomial-time algorithm.

\section{Problem Definition}

  In this paper, a graph always refers to an undirected graph that might have parallel edges and does not have loops.
	For a graph $G$, we denote by $V(G)$ and $E(G)$ the vertex set and edge set of $G$, respectively. 
	An edge subset $M\subseteq E$ is called a {\em matching} in $G$ if no two edges in $M$ share the end vertices.
	A matching $M$ is \emph{perfect} if $|M|=|V(G)|/2$.

  A graph is \emph{planar} if it can be drawn on the plane without edge crossing.
  Such a drawing is called a \emph{plane} drawing of the planar graph.
  A \emph{face} of a plane drawing is a maximal region of the plane that contains no point used in the drawing.
  There is a unique unbounded face, which is called the \emph{outer face}.
  A planar graph is \emph{outerplanar} if it has an \emph{outerplane} drawing, i.e., a plane drawing in which all vertices are incident to the outer face.

	For a matching $M$ in a graph $G$, a cycle $C$ in $G$ is called \emph{$M$-alternating} if 
  edges in $M$ and $E(G) \setminus M$ alternate in $C$.
  We identify a cycle with its edge set to simplify the notation. 
	We say that two perfect matchings $M$ and $N$ are \emph{reachable} (under the alternating cycle model) if there exists a sequence $\langle M_0, M_1, \ldots, M_t \rangle$ of perfect matchings in $G$ such that 
  \begin{enumerate}
	\item[(\one)] $M_0 = M$ and $M_t = N$;
	\item[(\two)] $M_{i}=M_{i-1} \symmdiff C_{i}$ for some $M_{i-1}$-alternating cycle $C_{i}$ for each $i = 1,\ldots, t$.
  \end{enumerate}
	Such a sequence is called a \emph{reconfiguration sequence} between $M$ and $N$, and its \emph{length} is defined as $t$.

	For two perfect matchings $M$ and $N$, the subgraph $M\symmdiff N$ consists of disjoint $M$-alternating cycles $C_1, \dots, C_t$.
	Thus it is clear that $M$ and $N$ are always reachable for any two perfect matchings $M$ and $N$ by setting $M_i=M_{i-1}\symmdiff C_i$ for $i=1,\dots,t$.
	In this paper, we are interested in finding a \emph{shortest} reconfiguration sequence of perfect matchings.
	That is, the problem is defined as follows:
	\begin{center}
		\parbox{0.95\hsize}{
		\begin{listing}{{\bf Input:}}
      \item[{\prob}]
			\item[{\bf Input:}] A graph $G$ and two perfect matchings $M$ and $N$ in $G$
			\item[{\bf Find:}] A shortest reconfiguration sequence between $M$ and $N$.
		\end{listing}}
	\end{center}
	We denote by a tuple $I=(G, M, N)$ an instance of \prob.
	Also, we denote by $\opt (I)$ the length of a shortest reconfiguration sequence of an instance $I$.
	We note that it may happen that $\opt (I)$ is much shorter than the number of disjoint $M$-alternating cycles in $M\symmdiff N$~(see Figure~\ref{fig:example}).


\section{Polynomial-Time Algorithm for Outerplanar Graphs}
\label{sec:outerplanar}

In this section, we prove that there exists a polynomial-time algorithm for {\prob} on an outerplanar graph, as follows.
\begin{theorem}\label{thm:outerplanar}
\prob on outerplanar graphs $G$ can be solved in $O(|V(G)|^5)$ time. 
\end{theorem}

We give such an algorithm in this section. 
Let $I=(G, M, N)$ be an instance of the problem such that $G=(V, E)$ is an outerplanar graph.
We first observe that it suffices to consider the case when $G$ is 2-connected. 

\begin{lemma}
  \label{lem:2conn}
  Let $I=(G,M,N)$ be an instance of \prob, and 
  $G_1, \dots, G_p$ be the $2$-connected components of $G$.
  Furthermore, for every $i=1, \dots, p$, let
  $I_i = (G_i, M \cap E(G_i), N \cap E(G_i))$ be an instance of \prob.
  Then, $\opt (I)=\sum_{i=1}^p \opt (I_i)$.
\end{lemma}
\begin{proof}
	Let $G_1, \dots, G_p$ be 2-connected components in $G$.
	Then, since any $M'$-alternating cycle is contained in some $G_i$ for a perfect matching $M'$ of $G$, it suffices to solve the problem for each $G_i$.
	Specifically, it holds that $\opt (I)=\sum_{i=1}^p \opt (I_i)$, where $I_i = (G_i, M\cap E(G_i), N\cap E(G_i))$.
\end{proof}

Since the 2-connected components of a graph can be found in linear time,
the reduction to 2-connected outerplanar graphs can be done in linear time, too.

We fix an outerplane drawing of a given $2$-connected outerplanar graph $G$, and
identify $G$ with the drawing for the sake of convenience.
We denote by $C_{\rm out}$ the outer face boundary.
Then $C_{\rm out}$ is a simple cycle since $G$ is $2$-connected.
We denote the set of the \emph{inner edges of $G$} by $E_{\rm in}=E\setminus C_{\rm out}$.
In other words, $E_{\rm in}$ is the set of chords of $C_{\rm out}$.

\subsection{Technical Highlight}
As mentioned in Introduction, there are two technical key points to develop a polynomial-time algorithm for \prob:
a lower bound on the length of a shortest reconfiguration sequence, and 
the characterization of unhappy moves. 
We here explain our ideas roughly, and will give detailed descriptions in the next subsections. 

Since $G$ is planar, we can define its ``dual-like'' graph $G^\ast$.
Then, $G^\ast$ forms a tree since $G$ is outerplanar and $2$-connected.   
(The definition of $G^\ast$ will be given in Section~\ref{subsec:outer_preliminary}, and an example is given in \figurename~\ref{fig:outerplanar_example2}.)
We make a correspondence between an edge in $G^\ast$ and a set of edges in $G$. 
Then, we will define the length $\ell (e^\ast)$ of each edge $e^\ast$ in $G^\ast$ so that it represents the ``gap'' between $M$ and $N$ when we are restricted to the edges in the corresponding set of $e^\ast$. 
It is important to notice that any cycle $C$ in $G$ corresponds to a subtree of $G^\ast$, and vice versa. 
Indeed, we focus on a cut $C^\ast$ of $G^\ast$ clipping the subtree from $G^\ast$, that is, the set of edges in $G^\ast$ leaving the subtree. 
If we apply an $M$-alternating cycle $C$ to a perfect matching $M$ of $G$, then it changes lengths $\ell (e^\ast)$ of the edges $e^\ast$ in the corresponding cut $C^\ast$. 

For our algorithm, we need a (good) lower bound for the length of a shortest reconfiguration sequence between two given perfect matchings $M$ and $N$.
Recall that $|M \symmdiff N|$ does not give a good lower bound under the alternating cycle model.
This is because we can take a cycle of an arbitrary (non-fixed) length, and hence $|M \symmdiff N|$ can decrease drastically by only a single alternating cycle. 
Furthermore, no matter how we define the length $\ell (e^\ast)$ of each edge $e^\ast$ in $G^\ast$, the total length of \emph{all} edges in $G^\ast$ does not give a good lower bound.
This is because a cycle $C$ of non-fixed length in $G$ may correspond to a cut $C^\ast$ having many edges in $G^\ast$, and hence it can change the total length drastically. 
Our key idea is to focus on the total length of each \emph{path} in $G^\ast$, that is, we take the \emph{diameter} of $G^\ast$ (with respect to length $\ell$) as a lower bound. 
Then, because $G^\ast$ is a tree, any path in $G^\ast$ can contain at most two edges from the corresponding cut $C^\ast$. 
Therefore, regardless of the cycle length,  the diameter of $G^\ast$ can be changed by only these two edges.
By carefully setting the length $\ell(e^\ast)$ as in
\eqref{eq:ell}, we will prove that the diameter of $G^\ast$
is not only a lower bound, but 
indeed gives the shortest length under the assumption that
$E_{\rm in} \cap M \cap N$ is empty. Therefore, the real difficulty
arises when $E_{\rm in} \cap M \cap N$ is not empty.

In the latter case, we will characterize the unhappy moves. 
Assume that we know the set $F \subseteq E_{\rm in} \cap M \cap N$ of
chords that are \emph{not} touched in a shortest reconfiguration sequence between $M$ and $N$; 
in other words, \emph{all} chords in $(E_{\rm in} \cap M \cap N) \setminus F$ must be touched for unhappy moves in that sequence. 
Then, we subdivide a given outerplanar graph $G$ into subgraphs $G_1, \ldots, G_{|F|+1}$ along the chords in $F$. 
Notice that each edge in $F$ appears on the outer face boundaries in two of these subgraphs.
Furthermore, each chord $e$ in these subgraphs satisfies $e \in (E_{\rm in} \cap M \cap N) \setminus F$ if $e \in M \cap N$.
Therefore, \emph{all} chords in these subgraphs are touched for unhappy moves as long as they are in $M\cap N$. 
Under this assumption, we will prove that the diameter of $G_i^*$ gives the shortest length of a reconfiguration sequence between $M \cap E(G_i)$ and $N \cap E(G_i)$.
Thus, we can solve the problem in polynomial time if we know $F$ which yields a shortest reconfiguration sequence between $M$ and $N$.  
Finally, to find such a set $F$ of chords, we construct a polynomial-time algorithm which employs a dynamic programming method along the tree $G^\ast$.

\subsection{Preliminaries: Constructing a Dual Graph} \label{subsec:outer_preliminary}
Let $I=(G, M, N)$ be an instance of \prob such that $G$ is a $2$-connected outerplanar graph.
Since $G$ is planar, we can define the \emph{dual} of $G$.
In fact, we here construct a graph $G^\ast$ obtained from the dual 
by applying a slight modification as follows. 
The construction is illustrated in \figurename~\ref{fig:outerplanar_example2}.
Let $V^\ast$ be the set of faces~(without the outer face) of $G$.
For a face $v^\ast\in V^\ast$, let $E_{v^\ast}$ be the set of edges around $v^\ast$.
We denote the set of faces touching the outer face by $U^\ast$, i.e., $U^\ast=\{v^\ast\in V^\ast\mid E_{v^\ast}\cap C_{\rm out}\neq \emptyset\}$.
We make a copy of $U^\ast$, denoted by $\tilde{U}^\ast$.
We set the vertex set of $G^\ast$ to be $V^\ast\cup \tilde{U}^\ast$.
For $v^\ast, w^\ast$ in $V^\ast$, an edge $v^\ast w^\ast$ in $G^\ast$ exists if and only if the faces $v^\ast$ and $w^\ast$ share an edge in $E_{\rm in}$, i.e., $|E_{v^\ast}\cap E_{w^\ast}|=1$.
Also $G^\ast$ has an edge between $u^\ast$ and $\tilde{u}^\ast$ for every $u^\ast\in U^\ast$.
Thus the edge set of $G^\ast$ is given by
\[
E(G^\ast) = \{v^\ast w^\ast \mid v^\ast, w^\ast \in V^\ast,\ |E_{v^\ast}\cap E_{w^\ast}|=1 \}\cup \{u^\ast \tilde{u}^\ast \mid u^\ast\in U^\ast\}.
\]
The first part is denoted by $E^\ast_{\rm in}$, and the second part is denoted by $\tilde{E}^\ast$.
We observe that $G^\ast$ is a tree, since $G$ is $2$-connected and outerplanar.
A face of $G$ that touches only one face~(other than the outer face) is called a \textit{leaf in $G^\ast - \tilde{U}^\ast$}.
We note that there is a one-to-one correspondence between edges in $E_{\rm in}$ of $G$ and $E^\ast_{\rm in}$ of $G^\ast$.
For an edge subset $F\subseteq E_{\rm in}$, $F^\ast$ denotes the corresponding edge subset in $G^\ast$, that is, $F^\ast = \{e^\ast \in E^\ast_{\rm in}\mid e\in F\}$.
Conversely, for an edge subset $F^\ast\subseteq E (G^\ast)$, $F$ denotes the corresponding edge subset in $E_{\rm in}$, that is, $F = \{e \in E_{\rm in}\mid e^\ast\in F^\ast\cap E^\ast_{\rm in}\}$.
We extend this correspondence to $\tilde{E}^\ast$, that is, $u^\ast \tilde{u}^\ast \in \tilde{E}^\ast$ corresponds to the edge set $E_{u^\ast} \cap C_{\rm out}$ for $u^\ast \in U^\ast$, and vice versa.

\begin{figure}[t]
  \centering
  \scalebox{0.8}{\includegraphics{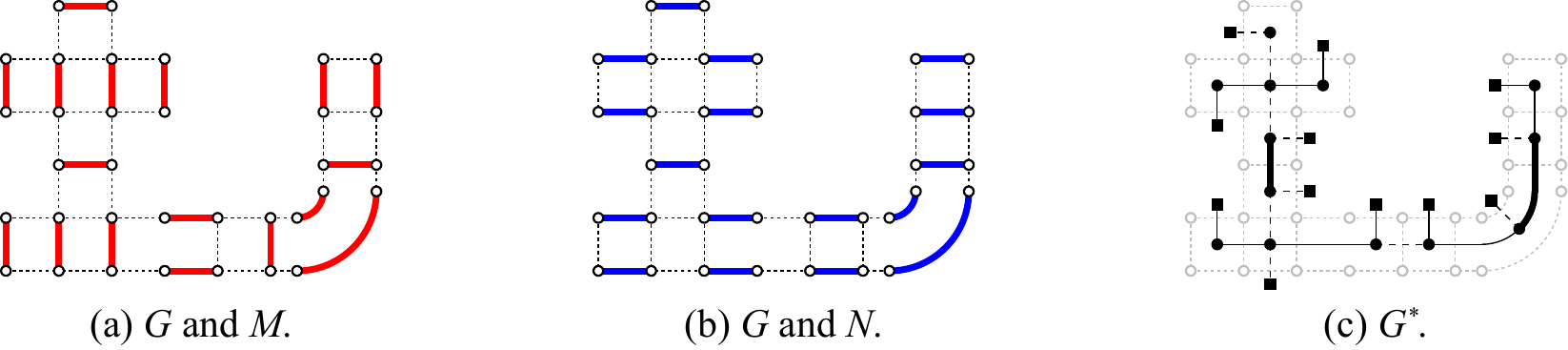}}
  \caption{The construction of $G^\ast$ and the length function $\ell$.
  In (c), the edge lengths are depicted by different styles:
  thick solid lines represent edges of length two,
  thin solid lines represent edges of length one, and
  dotted lines represent edges of length zero.}
  \label{fig:outerplanar_example2}
\end{figure}

It follows from the duality that there is a relationship between a cut in $G^\ast$ and a cycle in $G$. 
Suppose that we are given a cycle $C$ $(\neq C_{\rm out})$ in $G$.
Then, since $G$ is outerplanar, the cycle $C$ surrounds a set $X^\ast$ of faces such that $X^\ast$ does not have the outer face.
The set $X^\ast$ induces a connected graph (subtree) in $G^\ast$, and the set of edges leaving from $X^\ast$ yields a cut $C^\ast=\{e^\ast = v^\ast w^\ast \mid v^\ast\in X^\ast, w^\ast\in V(G^\ast)\setminus X^\ast\}$.
Conversely, let $X^\ast\subseteq V^\ast$ be a vertex subset of $G^\ast$ such that the subgraph induced by $X^\ast$ is connected.
Then the set of edges leaving from $X^\ast$ yields a cut $C^\ast$ in $G^\ast$, 
which corresponds to a cycle in $G$.

We classify faces in $U^\ast$ into two groups.
For a face $u^\ast$ in $U^\ast$, the edge set $E_{u^\ast}\cap C_{\rm out}$ forms a family $\mathcal{P}_{u^\ast}$ of disjoint paths.
Since $M$ and $N$ are perfect matchings, each path $P$ in $\mathcal{P}_{u^\ast}$ is both $M$-alternating and $N$-alternating. 
In addition, $P$ satisfies either 
\begin{enumerate}
\item[(\one)] $E(P) \subseteq M\symmdiff N$, or 
\item[(\two)] $(M \symmdiff N) \cap E(P) = \emptyset$.
\end{enumerate}
Furthermore, we observe that either (\one) holds for \emph{every} path $P$ in $\mathcal{P}_{u^\ast}$, or (\two) holds for \emph{every} path $P$ in $\mathcal{P}_{u^\ast}$.
Indeed, since $M\symmdiff N$ consists of disjoint cycles, if some path $P$ in $\mathcal{P}_{u^\ast}$ satisfies~(\one), then $P$ is included in 
a cycle $C$ in $M \symmdiff N$ that separates $u^\ast$ from the outer face.
Since the other paths in $\mathcal{P}_{u^\ast}$ touch the outer face, they are on $C$.
Thus every path satisfies~(\one), which shows the observation.
We divide $U^\ast$ into two groups $U^\ast_1$ and $U^\ast_2$ where each face in $U^\ast_1$ satisfies (\one) for every path, while each face in $U^\ast_2$ satisfies (\two) for every path.

For an edge $e^\ast$ in $E(G^\ast)$, we define the length $\ell (e^\ast)$ to be
\begin{equation}\label{eq:ell}
\ell (e^\ast) = 
\begin{cases}
|M\cap \{e\}|+|N\cap \{e\}| & \text{ if } e^\ast\in E^\ast_{\rm in}; \\
1 & \text{ if $e^\ast \in \tilde{E}^\ast$ is from $U^\ast_1$}; \\
0 & \text{ if $e^\ast \in \tilde{E}^\ast$ is from $U^\ast_2$}. 
\end{cases}
\end{equation}
See \figurename~\ref{fig:outerplanar_example2} for an example.
Let $\ell (u^\ast, v^\ast)$ be the length of the (unique) path from $u^\ast$ to $v^\ast$ in $G^\ast$.
We define the {\em gap} between $M$ and $N$ in the graph $G$ as the diameter of $G^\ast$, that is, we define 
\[
  {\rm gap}(I) = \max \{ \ell(u^\ast, v^\ast)\mid u^\ast, v^\ast\in V(G^\ast)\}.
\]
This value is simply denoted by ${\rm gap}(M, N)$ if $G$ is clear from the context.

\subsection{Characterization for the Disjoint Case}

Let $I=(G, M, N)$ be an instance of \prob such that $G$ is a $2$-connected outerplanar graph.
In this subsection, we show that if $E_{\rm in}\cap M\cap N$ is empty, 
we can characterize the optimal value with ${\rm gap}(I)$, which leads to a simple polynomial-time algorithm for this case.
We note that if $E_{\rm in}\cap M\cap N$ is empty, 
then no edge in $E_{\rm in}$ belongs to both $M$ and $N$, and hence $\ell (e^\ast)$ can only take $0$ or $1$;
in addition, ${\rm gap}(M,N) = 0$ if $M=N$.

\begin{lemma}\label{lem:diam_parity}
It holds that ${\rm gap}(M, N)$ is even.
\end{lemma}
\begin{proof}
Consider a path $P^\ast$ whose length is equal to ${\rm gap}(M, N)$ in $G^\ast$.
We may assume that the end vertices of $P^\ast$ are in $\tilde{U}^\ast$, as otherwise we can extend the path to some vertex in $\tilde{U}^\ast$ without decreasing the length.
Let $\tilde{u}, \tilde{v}\in \tilde{U}^\ast$ be the end vertices of $P^\ast$.
This means that the faces $u$ and $v$ touch the outer face.
Take arbitrary edges $e_u\in E_u\cap C_{\rm out}$ and $e_v\in E_v\cap C_{\rm out}$.
Then $(P\cap E_{\rm in})\cup \{e_u, e_v\}$ forms a cut $C$ in $G$ by the duality. 
By the definition of $\ell$, for $w\in\{u, v\}$, it holds that $\ell (w, \tilde{w})=0$ if and only if $|M\cap \{e_w\}|=|N\cap \{e_w\}|$.
Hence the parity of $\sum_{e^\ast\in E^\ast (P^\ast)}\ell (e^\ast)$ is the same as that of $|M\cap C|+|N\cap C|$.
Since $M$ and $N$ are perfect matchings, the parities of $|M\cap C|$ and $|N\cap C|$ are the same.
Therefore, $|M\cap C|+|N\cap C|$ is even, and thus ${\rm gap}(M, N)$ is also even.
\end{proof}

A main theorem of this subsection is to give a characterization of the optimal value with 
${\rm gap}(M, N)$.

\begin{theorem}\label{thm:subprob}
Let $I=(G, M, N)$ be an instance of \prob such that $G$ is a $2$-connected outerplanar graph.
If $E_{\rm in}\cap M\cap N$ is empty, then it holds that
$\opt (I)  = {{\rm gap}(M, N)}/2$.
\end{theorem}
\begin{proof}
To show the theorem, we first prove the following claim.
\begin{claim}
\label{clm:lb}
For any $M$-alternating cycle $C$, it holds that
\[{\rm gap}(M, N) \leq {\rm gap}(M\symmdiff C, N) +2.\]
\end{claim}
\begin{claimproof}[Proof of Claim \ref{clm:lb}]
By the duality, the cycle $C$ in $G$ corresponds to a cut $C^\ast$ in $G^\ast$ such that the inside is connected.
Such a cut intersects with any path in $G^\ast$ at most twice as $G^\ast$ is a tree, and only the intersected edges can change the length by one.
Therefore, the distance can be decreased by at most $2$.
\end{claimproof}

Consider a shortest reconfiguration sequence $\langle M_0, M_1, \ldots, M_t \rangle$ from $M_0=M$ to $M_t = N$.
Then, $t = \opt (I)$.
For each $i=1,\dots, t$,  it then holds that
${\rm gap}(M_{i-1}, N) \leq {\rm gap}(M_i, N)+2$.
By repeatedly applying the above inequalities, we obtain
\[
	{\rm gap}(M, N)= {\rm gap}(M_0, N) \leq {\rm gap}(M_t, N)+2t = 2t = 2 \opt(I)
\]
since ${\rm gap}(M_t, N)=0$.
Hence it holds that
$\opt(I) \geq {{\rm gap}(M, N)}/2$.

It remains to show that 
$\opt(I) \leq {{\rm gap}(M, N)}/2$.
We prove the following claim.
\begin{claim}\label{clm:hard_direction}
There exists an $M$-alternating cycle $C$ such that 
\begin{equation}\label{eq:someC}
{\rm gap}(M, N) = {\rm gap}(M\symmdiff C, N) +2.
\end{equation}
\end{claim}
\begin{claimproof}[Proof of Claim \ref{clm:hard_direction}]
We prove the claim by induction on the number of edges.

We first observe that we may assume that $E_{\rm in}\setminus (M\cup N)=\emptyset$.
Otherwise, we can just delete all the edges in $E_{\rm in}\setminus (M\cup N)$, and apply the induction to find an $M$-alternating cycle $C$ that satisfies \eqref{eq:someC} for the modified graph.
Since the deletion does not change the gap, $C$ is a desired cycle in $G$ as well.
Therefore, we may assume that all the edges in $E^\ast_{\rm in}$ have length $1$.

\begin{figure}[t]
	\centering
	\includegraphics{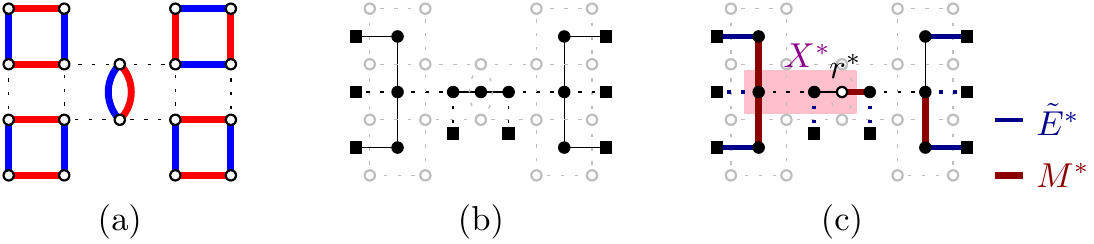}
	\caption{Illustration of the proof of Claim \ref{clm:hard_direction}.
		(a) The perfect matchings $M$ and $N$ are shown by red and blue, respectively.
		(b) The graph $G^\ast$.
		(c) The center $r$ and the chosen set $X^\ast$.}
	\label{fig:disjoint_case}
\end{figure}

In addition, we may assume that any leaf $u^\ast$ in $G^\ast - \tilde{U}^\ast$ belongs to $U^\ast_1$.
In other words, $M$ and $N$ are distinct in $E_{u^\ast}\cap C_{\rm out}$.
Indeed, suppose that there exists a leaf $u^\ast$ in $U^\ast_2$.
Then $\ell (u^\ast, \tilde{u}^\ast)=0$.
Since any chord is in either $M$ or $N$ by the above observation and the assumption that $E_{\rm in} \cap M \cap N = \emptyset$, $\ell (u^\ast, v^\ast)=1$, where $v^\ast$ is the unique neighbor to $u^\ast$ in $G^\ast - \tilde{U}^\ast$.
We delete $E_{u^\ast}\setminus E_{\rm in}$ from $G$, $M$, and $N$, and 
then delete all the isolated vertices. 
We denote the obtained graph by $G'$. 
This corresponds to deleting the face $u^\ast$ with $\tilde{u}^\ast$ from $G^\ast$, and adding $\tilde{v}^\ast$ to $G^\ast$ if necessary.
We can see that, in the modified graph $(G')^\ast$, we have $\ell (v^\ast, \tilde{v}^\ast)=1$, as $E_{u^\ast}\cap E_{v^\ast}$ is in either $M$ or $N$.
Hence this deletion preserves ${\rm gap}(M, N)$. 
We then apply the induction to $G'$ to find an $M$-alternating cycle $C$ that satisfies \eqref{eq:someC}.
This cycle is a desired one in $G$.
Thus we may assume that any leaf $u^\ast$ in $G^\ast - \tilde{U}^\ast$ belongs to $U^\ast_1$.

Since ${\rm gap}(M, N)$ is even by Lemma~\ref{lem:diam_parity}, we have ${\rm gap}(M, N) = 2d$ for some positive integer $d$. 
Then there exists a vertex $r^\ast\in V^\ast$ of $G^\ast$ such that, for every $v^\ast \in V(G^\ast)$, the $r^\ast$-$v^\ast$ path has length at most $d$.
Let $X^\ast\subseteq V^\ast$ be a minimal vertex subset of $G^\ast$ such that 
\begin{itemize}
\item $r^\ast\in X^\ast$
\item the subgraph induced by $X^\ast$ is connected in $G^\ast$
\item the cut $C^\ast=\{e^\ast = u^\ast v^\ast \mid u^\ast\in X^\ast, v^\ast\in V(G^\ast)\setminus X^\ast\}$ has only edges in $M^\ast \cup \tilde{E}^\ast$.
\end{itemize}
Such $X^\ast$ always exists
as $V^\ast$ satisfies all the conditions.
The cut $C^\ast$ corresponds to a cycle $C$ in $G$.
An example is given in \figurename~\ref{fig:disjoint_case}.

We claim that $C$ is $M$-alternating.
Assume not.
Then there exist two consecutive edges $e=uv$, $e'=v w$ in $C$ such that $e, e'\not\in M$, 
which implies that $e, e' \in C_{\rm out}$ as $E(C^\ast) \subseteq M^\ast \cup \tilde{E}^\ast$. 
Since $M$ is a perfect matching, the vertex $v$ is incident to another edge $f$ in $M$.
Since $G$ is 2-connected and outerplanar, there exists a path $P$ from $v$ to $C$ using the edge $f$ that is internally disjoint from $C$.
If $P$ has more than one edges, then $P$ ends with $u$ or $w$, since $G$ is outerplanar, which contradicts that 
$e, e' \in C_{\rm out}$.
Hence $P$ has only one edge.
However, this implies that $C$ has a chord in $M$, which contradicts that $C$ was chosen to be minimal.
Thus $C$ is an $M$-alternating cycle.

Consider taking $M\symmdiff C$.
Let $\ell'$ be the length defined by \eqref{eq:ell} with $M\symmdiff C$ and $N$.
It follows that, for an edge $e^\ast\in E(G^\ast)$, 
\[
\ell' (e^\ast) = 
\begin{cases}
\ell (e^\ast) & \text{ if } e^\ast\not\in C^\ast;\\
1-\ell (e^\ast) & \text{ if } e^\ast\in C^\ast. 
\end{cases}
\]
We will show that, for any vertex $\tilde{v}^\ast$ in $\tilde{U}^\ast$, we have $\ell'(r^\ast, \tilde{v}^\ast)\leq d-1$.
This proves the claim, as, for any two vertices $\tilde{u}^\ast, \tilde{v}^\ast$ in $\tilde{U}^\ast$, 
it holds that 
\[
\ell'(\tilde{u}^\ast, \tilde{v}^\ast)\leq \ell'(r^\ast, \tilde{u}^\ast) + \ell'(r^\ast, \tilde{v}^\ast) \leq 2d -2.
\]

Since $r^\ast\in X^\ast$ and no vertex in $\tilde{U}^\ast$ is in $X^\ast$, the $r^\ast$-$\tilde{v}^\ast$ path $P$ intersects with $C^\ast$ exactly once.
Hence the length of $P$ is changed by one by taking $M\symmdiff C$.
So, if $\ell (r^\ast, \tilde{v}^\ast)\leq d-2$, then $\ell'(r^\ast, \tilde{v}^\ast)\leq d-1$.
Thus
it suffices to consider the case when $\ell (r^\ast, \tilde{v}^\ast)\geq d-1$, i.e., $\ell (r^\ast, \tilde{v}^\ast)= d-1$ or $d$.

Assume that $\ell (v^\ast, \tilde{v}^\ast)= 0$, which implies that 
$v^\ast \in U^\ast_2$ and hence $v^\ast$ is not a leaf in $G^\ast - \tilde{U}^\ast$. 
In this case, there exists a leaf $u^\ast$ in $G^\ast - \tilde{U}^\ast$ such that $\ell (r^\ast, {u}^\ast)\ge \ell (r^\ast, {v}^\ast) + 1$. 
Since $u^\ast \in U^\ast_1$, we obtain 
\[
\ell (r^\ast, \tilde{u}^\ast) = \ell (r^\ast, {u}^\ast) + 1\ge \ell (r^\ast, {v}^\ast) + 2 = \ell(r^\ast, \tilde{v}^\ast) + 2 \ge d+1, 
\]
which is a contradiction. 

Thus, we may assume that $\ell (v^\ast, \tilde{v}^\ast)= 1$.
If the $r^\ast$-$\tilde{v}^\ast$ path $P$ intersects $C^\ast\cap M^\ast$, then the intersected cut edge has length 1, and hence we see that $\ell'(r^\ast, \tilde{v}^\ast) = \ell (r^\ast, \tilde{v}^\ast) -1 \leq d-1$.
Otherwise, that is, if $P$ intersects with $C^\ast\cap \tilde{E}^\ast$, then the intersected cut edge is $(v^\ast, \tilde{v}^\ast)$, and hence 
$\ell'(r^\ast, \tilde{v}^\ast) = \ell (r^\ast, \tilde{v}^\ast) -1 \leq d-1$.
Thus, 
$\ell'(r^\ast, \tilde{v}^\ast) \leq d-1$ in each case.
\end{claimproof}

For a perfect matching $M_{i-1}$ in $G$, it follows from Claim~\ref{clm:hard_direction} that there exists an $M_{i-1}$-alternating cycle $C_i$ such that
${\rm gap}(M_{i-1}, N) = {\rm gap}(M_{i-1}\symmdiff C_i, N) +2$.
Define $M_i = M_{i-1}\symmdiff C_i$, and repeat finding an alternating cycle satisfying the above equation.
The repetition ends when ${\rm gap}(M_{i}, N) = 0$, which means that $M_{i} =  N$ when $E_{\rm in}\cap M\cap N$ is empty.
The number of repetitions is equal to ${\rm gap}(M, N)/2$, and therefore, we have  
$\opt(I) \leq {{\rm gap}(M, N)}/2$.
Thus the proof is complete.
\end{proof}

\subsection{General Case}

Let $I=(G, M, N)$ be an instance of \prob such that $G$ is a $2$-connected outerplanar graph.
Define $E'_{\rm in} = E_{\rm in}\cap M \cap N$.
In this subsection, we deal with the general case, that is, $E'_{\rm in}$ is not necessarily empty. 
Then, there is a case when changing an edge in $E'_{\rm in}$ reduces the number of reconfiguration steps as in \figurename~\ref{fig:example}.
We call such a move an \textit{unhappy move}.
The key idea of our algorithm is to detect a set of edges necessary for unhappy moves.

Since $G$ is outerplanar and $2$-connected, 
any $F \subseteq E'_{\rm in}$ divides the inner region of $C_{\rm out}$ into $|F|+1$ parts $R_1, \dots , R_{|F|+1}$. 
For each $i=1, \dots , |F|+1$, 
let $G_i$ be the subgraph of $G$ consisting of all the vertices and the edges in $R_i$ and its boundary. 
Thus, each edge $e \in F$ appears on the outer face boundaries in two of these subgraphs.
See \figurename~\ref{fig:fig_partition_example}.
Let $\mathcal{G}_F = \{G_1,\dots, G_{|F|+1} \}$.
Note that each graph in $\mathcal{G}_F$ is outerplanar and $2$-connected. 
For each $H \in \mathcal{G}_F$, let $I_H = (H, M \cap E(H), N \cap  E(H))$.
We now show the following theorem.

\begin{figure}[tb]
  \centering
	\scalebox{0.8}{\includegraphics{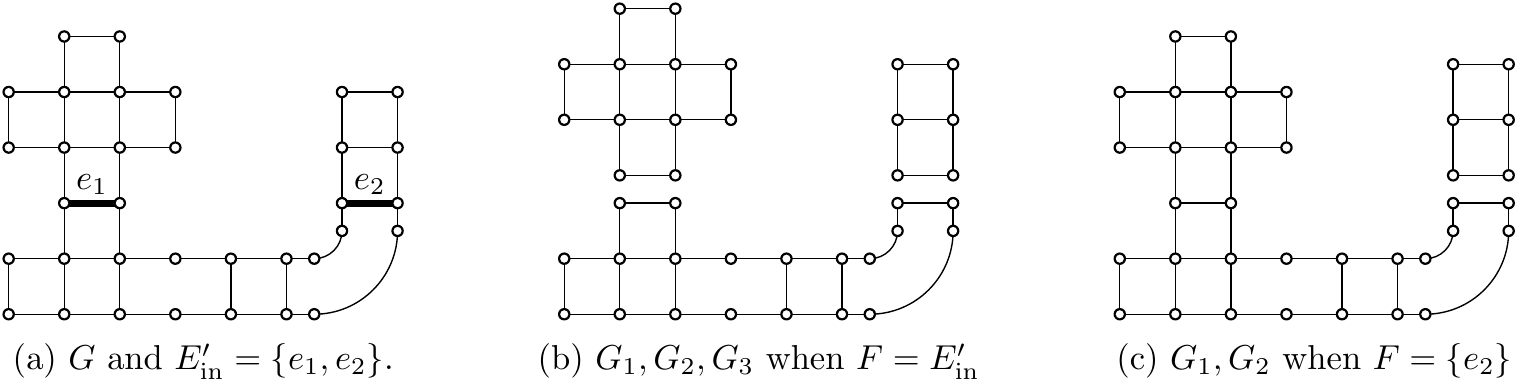}}
	\caption{The construction of a partition for the outerplanar graph in \figurename~\ref{fig:outerplanar_example2}. The edges in $E^\prime_{\mathrm{in}}$ are shown with bold lines.}
	\label{fig:fig_partition_example}
\end{figure}

\begin{theorem}
\label{thm:diamchar}
$\displaystyle \opt (I) = \frac{1}{2}  \min_{F\subseteq E'_{\rm in}} \sum_{H\in \mathcal{G}_F} {\rm gap}(I_H)$. 
\end{theorem}

\begin{proof}
	Let $\langle M_0, M_1, \ldots, M_t \rangle$ be a shortest reconfiguration sequence from $M_0=M$ to $M_t = N$.
	We denote by $C_i$ the $M_{i-1}$-alternating cycle with $M_i = M_{i-1}\symmdiff C_i$.
	Define 
  \[
  F_{\rm opt}=\{e\in E'_{\rm in} \mid e \not\in C_i, \forall i \},
  \]
  which is the set of edges in $E'_{\rm in}$ that do not touch in the shortest reconfiguration sequence.
	Then $C_i$ is contained in some $H\in \mathcal{G}_{F_{\rm opt}}$, and can be used to obtain a reconfiguration sequence from $M\cap E(H)$ to $N\cap E(H)$ in $H$.
	Therefore, we have
	\begin{equation}
	\label{eq:02}
	\opt (I) = \sum_{H\in \mathcal{G}_{F_{\rm opt}}} \opt (I_H).
	\end{equation}
	We can also see that 
	\begin{equation}
	\label{eq:01}
	\opt (I) \leq \sum_{H\in \mathcal{G}_F} \opt (I_H) 
	\end{equation}
	for any $F\subseteq E'_{\rm in}$.
	
	To evaluate $\opt (I_H)$ for $H\in \mathcal{G}_{F}$, 
	we slightly modify the instance $I_H$ 
	by replacing 
	every inner edge of $H$ contained in $M \cap N$
	by two parallel edges each in $M$ and $N$, respectively.
	The obtained graph is denoted by $H'$, and the corresponding instance is denoted by $I_{H'}$.
	Since a reconfiguration sequence for $I_{H'}$ can be converted to one for $I_H$, it holds that $\opt (I_H) \leq \opt (I_{H'})$, and hence
	\begin{equation}
	\label{eq:03}
	\opt (I) \leq \sum_{H\in \mathcal{G}_F} \opt (I_{H}) \leq \sum_{H\in \mathcal{G}_F} \opt (I_{H'}) 
	\end{equation}
	holds for any $F\subseteq E'_{\rm in}$ by (\ref{eq:01}).
	Moreover, by the definition of $F_{\rm opt}$, there exists an index $i$ such that $e\in C_i$ for any $e\in E'_{\rm in}\setminus F_{\rm opt}$. 
	Therefore, for $H\in \mathcal{G}_{F_{\rm opt}}$, 
	the shortest reconfiguration sequence for $I_H$ can be converted to a reconfiguration sequence for $I_{H'}$. 
	Thus, $\opt (I_H) \geq \opt (I_{H'})$ holds for $H\in \mathcal{G}_{F_{\rm opt}}$, and hence 
	\begin{equation}
	\label{eq:04}
	\opt (I) = \sum_{H\in \mathcal{G}_{F_{\rm opt}}} \opt (I_{H}) \ge \sum_{H\in \mathcal{G}_{F_{\rm opt}}} \opt (I_{H'})
	\end{equation}
	by (\ref{eq:02}). 
	By (\ref{eq:03}) and (\ref{eq:04}), we obtain
	\begin{equation}
	\label{eq:14}
	\opt (I) = \min_{F\subseteq E'_{\rm in}} \sum_{H\in \mathcal{G}_F} \opt (I_{H'}), 
	\end{equation}
	and $F_{\rm opt}$ is a minimizer of the right-hand side. 
	
	By (\ref{eq:14}) and Theorem~\ref{thm:subprob}, 
	we obtain 
	\begin{equation}
	\label{eq:05}
	\opt (I) =  \frac{1}{2} \min_{F\subseteq E'_{\rm in}} \sum_{H\in \mathcal{G}_F} {\rm gap}(I_{H'}), 
	\end{equation}
	because each $I_{H'}$ satisfies the condition in Theorem~\ref{thm:subprob}. 
	Since $(H')^*$ is obtained from $H^*$ by subdividing some edges of length two 
	into two edges of length one, the diameter of $(H')^*$ is equal to that of $H^*$, that is, ${\rm gap}(I_{H'})={\rm gap}(I_H)$. 
	Therefore, we obtain the theorem by~(\ref{eq:05}). 
\end{proof}

As an example, we apply this theorem to the instance in \figurename~\ref{fig:outerplanar_example2}.  
See \figurename~\ref{fig:fig_partition_example}(c).
If $F$ consists of only the right thick edge in \figurename~\ref{fig:outerplanar_example2}(c), then 
$\mathcal{G}_F$ consists two graphs $G_1$ and $G_2$ such that ${\rm gap}(I_{G_1}) = 6$ and ${\rm gap}(I_{G_2}) = 2$. 
Since we can check that such $F$ attains the minimum in the right-hand side of
Theorem \ref{thm:diamchar},
we obtain $\opt (I) = 4$ by Theorem~\ref{thm:diamchar}.

In order to compute the value 
in Theorem \ref{thm:diamchar}
efficiently, 
we reduce the problem to \TreeDiamDec, whose definition will be given later.

For $F \subseteq E'_{\rm in}$, 
let $F^*$ be the edge subset of $E^*_{\rm in}$ corresponding to $F$, 
and let $\mathcal{G}_F = \{G_1,\dots, G_{|F|+1} \}$. 
Then, $G^* - F^*$ consists of $|F|+1$ components $T_1, T_2, \dots , T_{|F|+1}$
such that $T_i$ coincides with $G^*_i$ (except for the difference of edges of length zero)
for $i=1, \dots , |F|+1$. 
In particular, for each $i$, we have 
${\rm gap}(I_{G_i}) = \max \{ \ell(u^\ast, v^\ast)\mid u^\ast, v^\ast\in V(T_i)\}$, 
where $\ell$ is the length function on $E(G^*)$ defined by the instance $I=(G, M, N)$. 
We call $\max \{ \ell(u^\ast, v^\ast)\mid u^\ast, v^\ast\in V(T_i)\}$ the \emph{diameter} of $T_i$, 
which is denoted by ${\rm diam}_{\ell}(T_i)$. 
Then, Theorem~\ref{thm:diamchar} shows that
\begin{equation}
\opt (I) = \frac{1}{2}  \min_{F\subseteq E'_{\rm in}} \sum^{|F|+1}_{i=1} {\rm diam}_{\ell}(T_i). \label{eq:redMSDD}
\end{equation}
Therefore, we can compute $\opt (I)$ by solving the following problem 
in which $T = G^*$ and $E_0 = (E'_{\rm in})^\ast$. 

	\begin{center}
		\parbox{0.95\hsize}{
		\begin{listing}{{\bf Input:}}
              \item[{\TreeDiamDec}] 
			\item[{\bf Input:}] A tree $T$, an edge subset $E_0 \subseteq E(T)$, and a length function $\ell:  E(T) \to \Zzero$. 
			\item[{\bf Find:}] An edge set $F \subseteq E_0$ that minimizes $\sum_{T'} {\rm diam}_{\ell}(T')$, 
where the sum is taken over all the components $T'$ of $T - F$. 
		\end{listing}}
	\end{center}

In the subsequent subsection, we show that \TreeDiamDec can be solved in time polynomial in $|V(T)|$ and $L := \sum_{e \in E(T)} \ell(e)$.

\begin{theorem}\label{thm:TreeDiamDec}
\TreeDiamDec  can be solved in $O(|V(T)|L^4)$ time, where $L := \sum_{e \in E(T)} \ell(e)$. 
\end{theorem}

Since (\ref{eq:redMSDD}) shows that 
\prob on outerplanar graphs is reduced to \TreeDiamDec in which $L=O(|V(T)|)$, 
we obtain Theorem~\ref{thm:outerplanar}. 


\subsection{Algorithm for \TreeDiamDec}

The remaining task is to show Theorem~\ref{thm:TreeDiamDec}, that is, to give an algorithm for \TreeDiamDec 
that runs in $O(|V(T)|L^4)$ time. 
For this purpose, we adopt a dynamic programming approach. 

We choose an arbitrary vertex $r$ of a given tree $T$, and regard $T$ as a rooted tree with the root $r$. 
For each vertex $v$ of $T$, we denote by $T_v$ the subtree of $T$ which is rooted at $v$ and is induced by all descendants of $v$ in $T$. 
(See \figurename~\ref{fig:subtree}(a).) 
Thus, $T=T_r$ for the root $r$. 
Let $w_1, w_2, \dots, w_q$ be the children of $v$, ordered arbitrarily. 
For each $j \in \{1, 2, \dots, q \}$, we denote by 
$T_v^j$ the subtree of $T$ induced by 
$\{v\} \cup V(T_{w_1}) \cup V(T_{w_2}) \cup \cdots \cup V(T_{w_j})$.
For example, in \figurename~\ref{fig:subtree}(b), the subtree $T_v^j$ is 
surrounded by a thick dotted rectangle.
For notational convenience, we denote by 
$T_v^0$ the tree consisting of a single vertex $v$. 
Then, $T_v = T_v^0$ for each leaf $v$ of $T$. 
Our algorithm computes and extends partial solutions for subtrees $T_v^j$ from the leaves to the root $r$ of $T$ by keeping the information required for computing (the sum of) diameters of a partial solution. 

\begin{figure}[tb]
  \centering
		\includegraphics[width=0.6\linewidth]{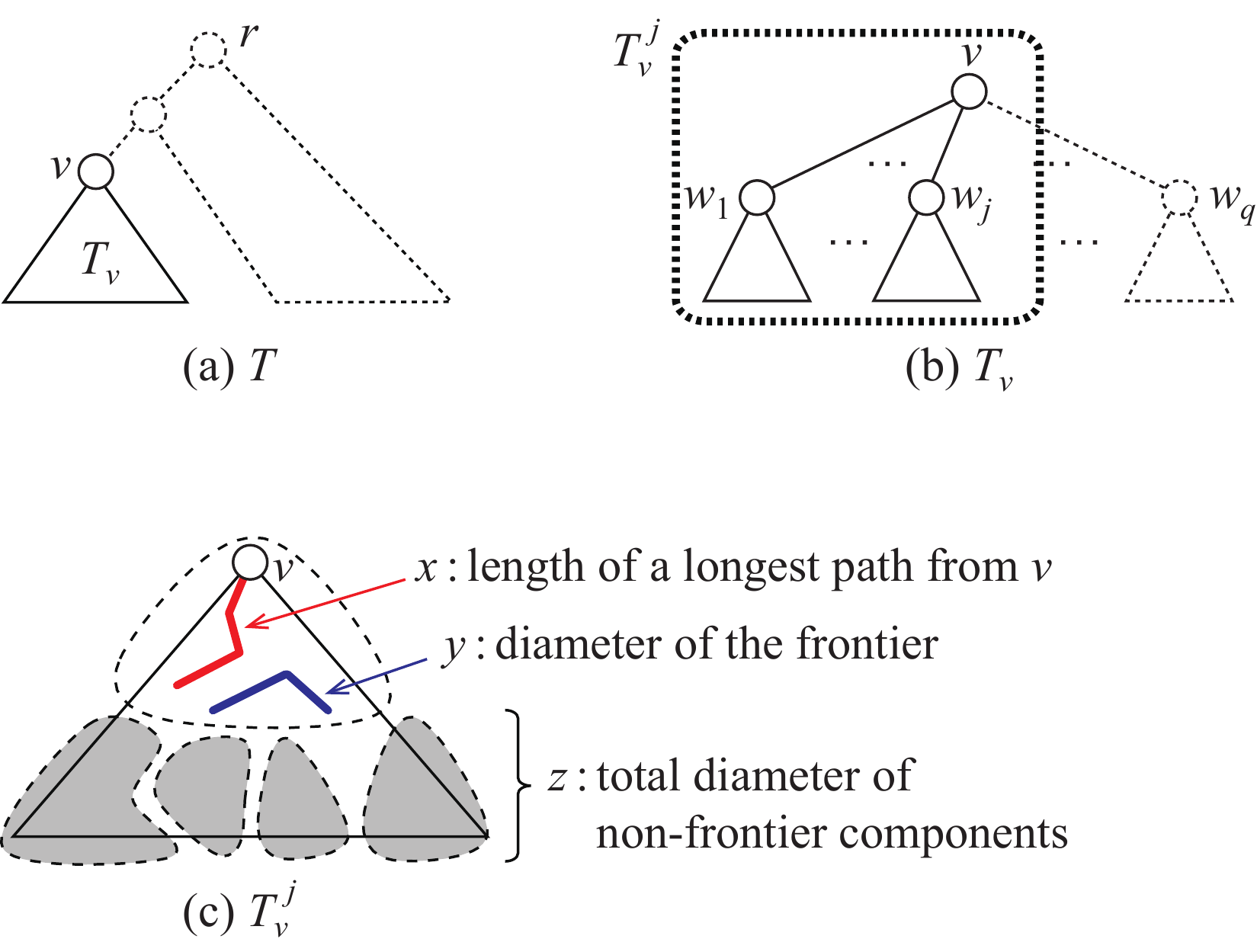}
	  \caption{(a) Subtree $T_v$ in the whole tree $T$, 
		(b) subtree $T_v^j$ in $T_v$, and
		(c) an $(x,y,z)$-separator of $T_v^j$.}
	\label{fig:subtree}
\end{figure}

We now define partial solutions for subtrees.
For a subtree $T_v^j$ and an edge subset $F^\prime \subseteq E_0 \cap E(T_v^j)$,
the \emph{frontier} for $F^\prime$ is the component (subtree) in $T_v^j - F^\prime$ that contains the root $v$ of $T_v^j$. 
We sometimes call it the \emph{$v$-frontier} for $F^\prime$ to emphasize the root $v$.
For three integers $x,y,z \in \{0,1,\dots,L\}$, 
the edge subset $F^\prime$ is called an \emph{$(x,y,z)$-separator} of $T_v^j$ if the following three conditions hold.
(See also \figurename~\ref{fig:subtree}(c).)
\begin{itemize}
\item $x = \max \{ \ell(v, u) \mid u \in V(T_{F^\prime})\}$, where $T_{F^\prime}$ is the $v$-frontier for $F^\prime$.
That is, the longest path from $v$ to a vertex in $T_{F^\prime}$ is of length $x$. 
\item $y = {\rm diam}_{\ell}(T_{F^\prime})$, that is, $y$ denotes the diameter of the $v$-frontier $T_{F^\prime}$ for $F^\prime$. 
\item $z = \sum_{T'} {\rm diam}_{\ell}(T')$, where the sum is taken over all the components $T'$ of $(T - F^\prime) \setminus T_{F^\prime}$. 
\end{itemize}
Note that $x \le y$ always holds for an $(x,y,z)$-separator of $T_v^j$. 
We then define the following function:
for a subtree $T_v^j$ and two integers $x,y \in \{0,1,\ldots, L\}$, we let 
\[
	f(T_v^j; x, y) = \min \left\{ z \mid \mbox{$T_v^j$ has an $(x,y,z)$-separator} \right\}. 
\]
Note that $f(T_v^j; x, y)$ is defined as $+ \infty$ 
if $T_v^j$ does not have an $(x,y,z)$-separator for any $z \in \{0,1,\ldots, L\}$. 
Then, the optimal objective value to \TreeDiamDec can be computed as $\min \{ y+f(T; x, y) \mid x,y \in \{0,1,\ldots, L\} \}$.

	For a given tree $T$, our algorithm computes $f(T_v^j; x,y)$ for all possible triplets $(T_v^j, x,y)$ from the leaves to the root $r$ of $T$, as follows.

\medskip

\noindent
{\bf Initialization.}
We first compute $f(T_v^0; x,y)$ for all vertices $v \in V(T)$ (including internal vertices in $T$). 
Recall that $T_v^0$ consists of a single vertex $v$. 
Therefore, we have 
\[
f(T_v^0; x,y) = 
\begin{cases}
0 & \mbox{if $x=y=0$};\\
+\infty  & \mbox{otherwise}.
\end{cases}
\]
Notice that we have computed $f(T_v; x,y)$ for all leaves $v$ of $T$, since $T_v =T_v^0$ if $v$ is a leaf. 
\medskip

\noindent
{\bf Update.}
We now consider the case where $j \ge 1$. 
To compute $f(T_v^j; x,y)$, we  classify $(x,y,z)$-separators of $T_v^j$ into the following two groups (a) and (b).
Note that $(x,y,z)$-separators of Group~(b) exist only when $vw_j \in E_0$. 
\medskip

\noindent
(a) The vertices $v$ and $w_j$ are contained in the same component. (See also \figurename~\ref{fig:treeupdate}(a).)

In this case, the edge $vw_j$ is not deleted, and the $v$-frontier for an $(x, y, z)$-separator of $T_v^j$ contains both $v$ and $w_j$. 
Therefore, we can obtain the $v$-frontier for an $(x, y, z)$-separator of $T_v^j$ by merging the $v$-frontier for some $(x^\prime, y^\prime, z^\prime)$-separator of $T_v^{j-1}$ with the $w_j$-frontier for some $(x^{\prime\prime}, y^{\prime\prime}, z^{\prime\prime})$-separator of $T_{w_j}$. 
Thus, we define
\[
f^{\textup{a}}(T_v^j;x,y) := \min \left\{ f(T_v^{j-1}; x^\prime, y^\prime) + f(T_{w_j}; x^{\prime\prime}, y^{\prime\prime}) \right\},
\]
where the minimum is taken over all integers $x^\prime, y^\prime, x^{\prime\prime}, y^{\prime\prime} \in \{0,1,\ldots,L\}$ such that
$x = \max \{ x^\prime, x^{\prime\prime}+\ell(vw_j) \}$ and
$y = \max \{ y^\prime, y^{\prime\prime}, x^\prime+\ell(vw_j)+x^{\prime\prime} \}$.  
\medskip

\noindent
(b) The vertices $v$ and $w_j$ are contained in different components. (See also \figurename~\ref{fig:treeupdate}(b).)

In this case, the edge $vw_j$ is deleted, and hence this case happens only when $vw_j \in E_0$. 
Then, the $v$-frontier for an $(x, y, z)$-separator of $T_v^j$ is the $v$-frontier for some $(x^\prime, y^\prime, z^\prime)$-separator of $T_v^{j-1}$. 
Note that $w_j$ is contained in a non-frontier component for the $(x, y, z)$-separator of $T_v^j$, but the component forms the $w_j$-frontier for some $(x^{\prime\prime}, y^{\prime\prime}, z^{\prime\prime})$-separator of $T_{w_j}$, as illustrated in \figurename~\ref{fig:treeupdate}(b).
Thus, we need to take the diameter of the $w_j$-frontier into account when we compute $f(T_v^j; x,y)$ from $f(T_v^{j-1}; x^\prime, y^\prime)$ and $f(T_{w_j}; x^{\prime\prime}, y^{\prime\prime})$. 
Therefore, we define
\[
f^{\textup{b}}(T_v^j;x,y) := \min \left\{ f(T_v^{j-1}; x^\prime, y^\prime) + f(T_{w_j}; x^{\prime\prime}, y^{\prime\prime}) + y^{\prime\prime} \right\},
\]
where the minimum is taken over all integers $x^\prime, y^\prime, x^{\prime\prime}, y^{\prime\prime} \in \{0,1,\ldots,L\}$ such that
$x = x^\prime$ and
$y = y^\prime$. 
\medskip

\begin{figure}[tb]
  \centering
		\includegraphics[width=0.9\linewidth]{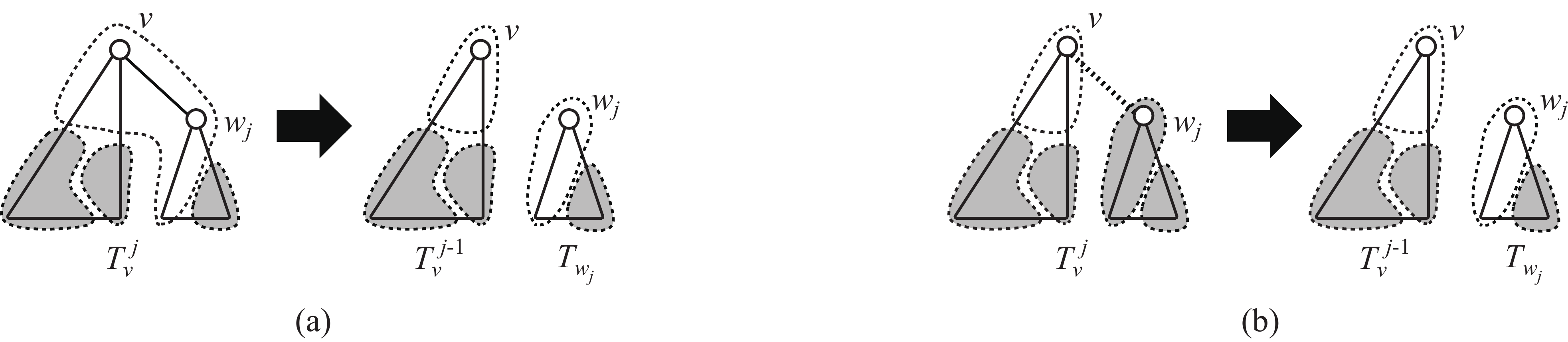}
	\caption{$(x,y,z)$-separators of a subtree $T_v^j$, and their restrictions to subtrees $T_v^{j-1}$ and $T_{w_j}$.}
	\label{fig:treeupdate}
\end{figure}

Then, we can compute $f(T_v^j; x,y)$ as follows:
\[
f(T_v^j; x,y) = 
\begin{cases}
\min \left\{ f^{\textup{a}}(T_v^j; x,y), f^{\textup{b}}(T_v^j; x,y) \right\} & \mbox{if $vw_j \in E_0$}; \\
f^{\textup{a}}(T_v^j; x,y) & \mbox{otherwise}.
\end{cases}
\]
Since $x^\prime, y^\prime, x^{\prime\prime}, y^{\prime\prime} \in \{0,1,\ldots,L\}$, this update can be done in $O(L^4)$ time for each subtree $T_v^j$. 
The number of subtrees $T_v^j$ is equal to $|V(T)| + |E(T)| = 2|V(T)|-1$. 
Therefore, this algorithm runs in $O(|V(T)| L^4)$ time in total.

Note that we can easily modify the algorithm so that 
we obtain not only the optimal value but also an optimal solution. 
This completes the proof of Theorem~\ref{thm:TreeDiamDec}. 

We remark here that 
the algorithm can be modified so that the running time is bounded by a polynomial in $|V(T)|$
by replacing the domain $\{0, 1, \dots , L \}$ of $x$ and $y$ with $D := \{ \ell(u, v) \mid u, v \in V(T)\}$. 
This modification is valid, because $f(T_v^j; x,y) = +\infty$ unless $x, y \in D$. 
Since $|D| =O(|V(T)|^2)$, the modified algorithm runs in $O(|V(T)| |D|^4) = O(|V(T)|^9)$ time. 
Note that, although this bound is polynomial only in $|V(T)|$, it is worse than $O(|V(T)| L^4)$ when $L=O(|V(T)|)$.


\section{NP-Hardness for Planar Graphs and Bipartite Graphs}
\label{sec:hardness}

In this section, we prove that {\prob} is \NP-hard even when the input graph is planar or bipartite.
\begin{theorem}
  \label{thm:hardness_planar}
	\prob is {\NP}-hard even for planar graphs of maximum degree three.
\end{theorem}

	We reduce the \ham problem, which is known to be {\NP}-complete even when a given graph is $3$-regular and planar~\cite{DBLP:journals/siamcomp/GareyJT76}.
	\begin{center}
		\parbox{0.95\hsize}{
		\begin{listing}{{\bf Question:} }
      \item[{\ham}]
			\item[{\bf Input:}] A $3$-regular planar graph $H=(V, E)$
			\item[{\bf Question:}] Decide whether $H$ has a Hamiltonian cycle, i.e., a cycle that goes through all the vertices exactly once.
		\end{listing}}
	\end{center}

\begin{proof}
	Let $H$ be a $3$-regular planar graph, which is an instance of 
  \ham.
	For each vertex $v\in V(H)$, we define a $8$-vertex graph $D_v$ (see also the top right in \figurename~\ref{fig:Reduction}):
	\begin{align*}
		V(D_v) & = \{v_1,v_2,v_3,v_4,v_5,v_6,v_7,v_8\},\\
		E(D_v) & = \{v_1 v_2, v_2 v_3, v_3 v_4, v_4 v_1, v_4 v_5, v_5 v_7, v_3 v_6, v_6 v_8\}.
	\end{align*}

	We construct an instance $I=(G, M, N)$ of our problem as follows.
 (See \figurename~\ref{fig:Reduction} as an example.)
	We subdivide each edge $e=uv$ in $H$ twice, and the obtained vertices are denoted by $u_e$ and $v_e$, where $u_e$ is closer to $u$.
	Then, for each vertex $v\in V(H)$, we replace $v$ with the graph $D_v$, and connect $v_7$ to $v_{e^{(1)}_v}$ and $v_{e^{(2)}_v}$, $v_8$ to $v_{e^{(2)}_v}$ and $v_{e^{(3)}_v}$, where $e^{(1)}_v$, $e^{(2)}_v$, $e^{(3)}_v$ are edges incident to $v$ and the order follows the planar embedding of $H$. 
	Let $E_v = \{v_7 v_{e^{(1)}_v}, v_7 v_{e^{(2)}_v}, v_8 v_{e^{(2)}_v}, v_8 v_{e^{(3)}_v}\}$.
	The resulting graph is denoted by $G$, i.e., $G$ is defined as follows:
	\begin{align*}
		V(G) & = \bigcup_{v\in V(H)} V(D_v) \cup \bigcup_{e=uv \in E(H)} \{u_e, v_e\},\\
		E(G) & = \left(\bigcup_{v\in V(H)}E (D_v)\cup E_v\right) \cup \{u_e v_e\mid e\in E(H)\}.
	\end{align*}
  It follows that $G$ is a planar graph of maximum degree three.
	Furthermore, we define initial and target perfect matchings $M$ and $N$ in $G$, respectively, to be 
	\begin{align*}
		M &= \{v_1 v_2, v_3 v_4, v_5 v_7, v_6 v_8 \mid v\in V(H)\}\cup \{u_e v_e \mid e\in E(H)\},\\
		N &= \{v_1 v_4, v_2 v_3, v_5 v_7, v_6 v_8 \mid v\in V(H)\}\cup \{u_e v_e \mid e\in E(H)\}.
	\end{align*}
	This completes the construction of our corresponding instance $I = (G,M,N)$. 
	The construction can be done in polynomial time. 

\begin{figure}
  \centering
		\includegraphics{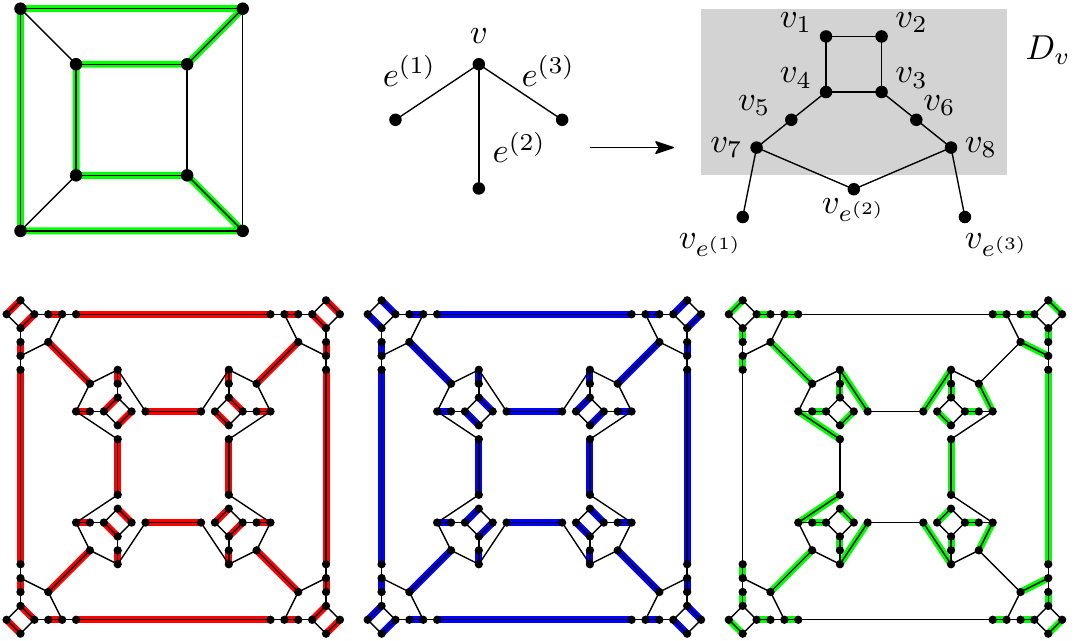}
	\caption{Reduction for planar graphs of maximum degree three. Top left: a yes instance $H$ of \ham with a green Hamiltonian cycle. Top right: the constructed fragment $D_v$.
		Bottom left: The initial perfect matching $M$ (red).
		Bottom middle: The target perfect matching $N$ (blue).
		Bottom right: The perfect matching obtained as $M \symmdiff C$, where
		$C$ corresponds to the Hamiltonian cycle of $H$.
	}
	\label{fig:Reduction}
\end{figure}

	We then give the following claims. 
  Recall that $t^\ast$ is the length of a shortest reconfiguration
  sequence for the constructed instance $I$.

  \begin{claim}\label{clm:hardness1}
	  It holds that $t^\ast\geq 2$.
  \end{claim}
  \begin{claimproof}[Proof of Claim \ref{clm:hardness1}]
	  We observe that, if $t^\ast=1$, then $M \symmdiff N$ must consist of one $M$-alternating cycle, but it is not true for our instance $I$.
	  Thus the length of a reconfiguration sequence is at least two.
  \end{claimproof}
    
	We remark that $G$ has an $M$-alternating path from $v^{(x)}_e$ to $v^{(y)}_e$ for any $x, y\in\{1,2,3\}$ with $x\neq y$.
	This implies that, for a cycle $C$ in $H$, there exists a corresponding $M$-alternating cycle $C'$ in $G$ such that it goes through vertices of $D_v$ for every $v\in V(C)$ and edges $u_e v_e$ for every $e\in E(C)$.
	
  \begin{claim}\label{clm:hardness2}
	  If $H$ has a Hamiltonian cycle $C$, then it holds that $t^\ast= 2$.
  \end{claim}
  \begin{claimproof}[Proof of Claim \ref{clm:hardness2}]
	  We see that $G$ has an $M$-alternating cycle $C'$, corresponding to $C$ of $H$, that has one edge $v_3 v_4$ of $C_v$ for each $v\in V(C)$.
	  Then $M^\prime=M\symmdiff C'$ is a perfect matching.
	  In a similar way, $G$ has an $M^\prime$-alternating cycle $C''$, corresponding to $C$, that uses three edges $v_3 v_2$, $v_2 v_1$, and $v_1 v_4$ of $C_v$ for each $v\in V(C)$.
	  Then $M^\prime\symmdiff C''$ is equal to $N$.
	  Thus we can find a reconfiguration sequence of length two, which is shortest by Claim~\ref{clm:hardness1}.
  \end{claimproof}
  
  The $4$-cycle formed by $v_1,v_2,v_3,v_4$ is denoted by $C_v$.
  
  \begin{claim}\label{clm:hardness3}
	  If $t^\ast = 2$, then $H$ has a Hamiltonian cycle.
  \end{claim}
  \begin{claimproof}[Proof of Claim \ref{clm:hardness3}]
	  We denote by $\langle M, M', N \rangle$ a shortest reconfiguration sequence of $I$.
	  Let $C = M\symmdiff M^\prime$.
	  We may assume that $C$ is not $C_v$ for any $v\in H$, as $t^\ast = 2$.
	  We will prove that the edge subset $F=\{e \in E(H)\mid u_e v_e\in C\}$ forms a Hamiltonian cycle in $H$.
	  We denote $W_C$ by the set of vertices in $H$ used in $F$.
	  Let $\overline{W_C}=V(H)\setminus W_C$.
	  Since $M^\prime\cap C_v$ and $N\cap C_v$ are distinct for $v\in \overline{W_C}$, the symmetric difference $M^\prime\symmdiff N$ has at least $|\overline{W_C}|$ disjoint $M^\prime$-alternating cycles.
	  Moreover, for a vertex $v\in W_C$, we see that $M^\prime\cap C_v=\{v_1 v_2\}$ and $N\cap C_v=\{v_1 v_4, v_2 v_3\}$, that are distinct.
	  Hence $M^\prime\symmdiff N$ has at least one $M^\prime$-alternating cycle disjoint from $\bigcup_{v\in \overline{W_C}} V(D_v)$.
	  Therefore, we have at least $|\overline{W_C}|+1$ disjoint $M^\prime$-alternating cycles.
	  However, $M^\prime\symmdiff N$ must consist of one cycle~(see Claim~\ref{clm:hardness1}), implying that $\overline{W_C}=\emptyset$.
	  This means that $C$ goes through $C_v$ for every $v$, and hence $C'$ is a Hamiltonian cycle in $H$.
	  Thus the claim holds.
  \end{claimproof}
	Therefore, it follows that $H$ has a Hamiltonian cycle if and only if $t^\ast = 2$.
	This completes the proof of Theorem~\ref{thm:hardness_planar}.
\end{proof}

The hardness for bipartite graphs of maximum degree at most three
can be obtained with a similar proof.

\begin{theorem}\label{thm:hardness_bipartite}
	{\prob} is \NP-hard even for bipartite graphs of maximum degree three.
\end{theorem}
We reduce the directed Hamiltonian cycle problem, which is known to be
\NP-complete even if digraphs have the maximum in-degree two
and the maximum out-degree two~\cite{DBLP:journals/ipl/Plesnik79}.
	\begin{center}
		\parbox{0.95\hsize}{
		\begin{listing}{{\bf Question:}}
      \item[{\diham}]
			\item[{\bf Input:}] A digraph $H=(V, E)$
			\item[{\bf Question:}] Decide whether $H$ has a directed Hamiltonian cycle, i.e., a directed cycle that goes through all the vertices exactly once.
		\end{listing}}
	\end{center}

\begin{proof}
Let $H$ be a digraph, which is an instance of the
directed Hamiltonian cycle problem.
We assume that $|V(H)|\geq 3$; otherwise the problem is trivial.
For each vertex $v \in V(H)$, we define a $6$-vertex graph $D_v$ (see the
top right in \figurename~\ref{fig:Reduction_bipartite}):
\begin{align*}
V(D_v) &= \{v^+, v^-, v_1, v_2, v_3, v_4, v_5, v_6\},\\
E(D_v) &= \{v^+v_1, v_1v_2, v_2v_3, v_3v_4, v_4v_5, v_5v_2, v_5v_6, v_6v^-\}.
\end{align*}
The cycle of length four formed by $v_2, v_3, v_4, v_5$ is denoted by $C_v$.

\begin{figure}[t]
\centering
\includegraphics{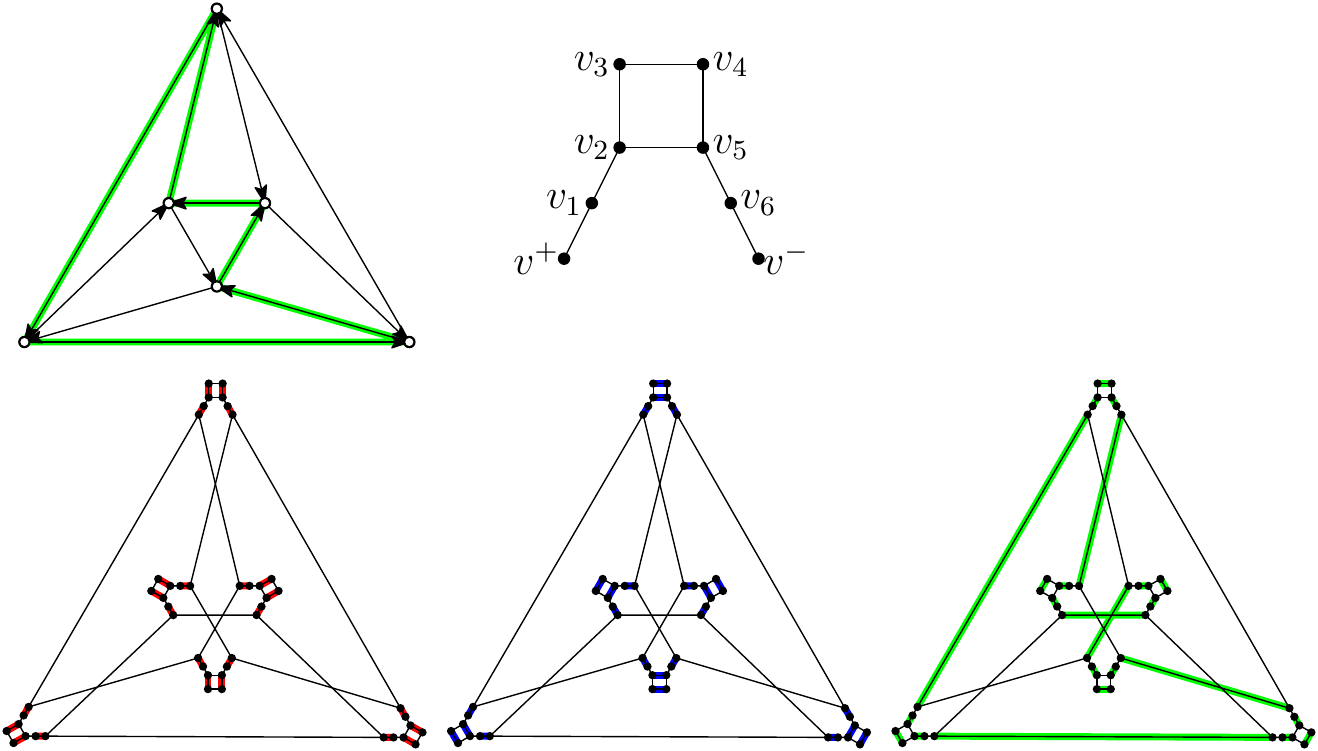}
\caption{Reduction for bipartite graphs of maximum degree three. Top left: a yes instance $H$ of \diham with a green directed Hamiltonian cycle. Top middle: the constructed fragment $D_v$.
Bottom left: The initial perfect matching $M$ (red).
Bottom middle: The target perfect matching $N$ (blue).
Bottom right: The perfect matching obtained as $M \symmdiff C$ where
$C$ corresponds to the directed Hamiltonian cycle of $H$.
}
\label{fig:Reduction_bipartite}
\end{figure}

We construct an instance $I=(G, M, N)$ of our problem
as follows.
The vertex set and the edge set of $G$ are defined as
\[
V(G) = \bigcup_{v \in V(H)} V(D_v),
\qquad
E(G) = \bigcup_{v \in V(H)} E(D_v) \cup \{ u^-v^+ \mid uv \in E(H) \},
\]
respectively.
Namely, for each directed edge from $u$ to $v$ in $H$,
we add an undirected edge to $G$ between $u^-$ and $v^+$.
This finishes the construction of $G$.
Note that $G$ is bipartite and its maximum degree is at most three as
both the maximum in-degree and the maximum out-degree of $H$ are at most two.
Let $M$ and $N$ be defined as
\begin{align*}
M &= \bigcup_{v \in V(H)} \{ v^+v_1, v_2v_3, v_4v_5, v_6v^-\},\\
N &= \bigcup_{v \in V(H)} \{ v^+v_1, v_2v_5, v_3v_4, v_6v^-\}.
\end{align*}
Refer to \figurename~\ref{fig:Reduction_bipartite} for the illustration.
Let $t^\ast$ be the length of a shortest reconfiguration sequence for $I$.

\begin{claim}\label{clm:hardness1-bip}
It holds that $t^\ast \geq 2$.
\end{claim}
\begin{proof}
If $t^\ast=1$, then $M \symmdiff N$ must consist of one $M$-alternating cycle, but this is not the case for our instance $I$.
Thus, the length of a reconfiguration sequence is at least two.
\end{proof}

\begin{claim}\label{clm:hardness2-bip}
	If $H$ has a directed Hamiltonian cycle $C$, then it holds that $t^\ast= 2$.
\end{claim}
\begin{proof}
	We see that $G$ has an $M$-alternating cycle $C'$, corresponding to $C$ of $H$, that has four edges $v^+ v_1, v_2 v_3, v_4 v_5, v_6 v^-$ of $D_v$ for each $v\in V(C)$.
	Then $M^\prime=M\symmdiff C'$ is a perfect matching.
	In a similar way, $G$ has an $M^\prime$-alternating cycle $C''$, corresponding to $C$, that uses three edges $v^+ v_1$, $v_2 v_5$, and $v_6 v^-$ of $C_v$ for each $v\in V(C)$.
	Then $M^\prime\symmdiff C''$ is equal to $N$.
	Thus we can find a reconfiguration sequence of length two, which is the shortest by Claim~\ref{clm:hardness1-bip}.
\end{proof}

\begin{claim}\label{clm:hardness3-bip}
	If $t^\ast = 2$, then $H$ has a directed Hamiltonian cycle.
\end{claim}
\begin{proof}
	Let $\langle M, M', N \rangle$ be a shortest reconfiguration sequence of $I$.
  Let $C =M \symmdiff M^\prime$.
	We may assume that $C$ is not $C_v$ for any $v\in H$, as $t^\ast = 2$.
	We will prove that the edge subset $F=\{uv \in E(H)\mid u^- v^+\in C\}$ forms a Hamiltonian cycle in $H$.
	We denote $W_C$ by the set of vertices in $H$ used in $F$.
	Let $\overline{W_C}=V(H)\setminus W_C$.
	Since $M^\prime\cap C_v$ and $N\cap C_v$ are distinct for $v\in \overline{W_C}$, the symmetric difference $M^\prime\symmdiff N$ has at least $|\overline{W_C}|$ disjoint $M^\prime$-alternating cycles.
	Moreover, for a vertex $v\in W_C$, we see that $M^\prime\cap C_v=\{v_3 v_4\}$ and $N\cap C_v=\{v_3 v_4, v_2 v_5\}$, that are distinct.
	Hence $M^\prime\symmdiff N$ has at least one $M^\prime$-alternating cycle disjoint from $\bigcup_{v\in \overline{W_C}} V(D_v)$.
	Therefore, we have at least $|\overline{W_C}|+1$ disjoint $M^\prime$-alternating cycles.
	However, $M^\prime\symmdiff N$ must consist of one cycle~(see Claim~\ref{clm:hardness1-bip}), implying that $\overline{W_C}=\emptyset$.
	This means that $C$ goes through $C_v$ for every $v$, and hence $C'$ is a Hamiltonian cycle in $H$.
	Thus the claim holds.
\end{proof}
	Therefore, it follows that $H$ has a directed Hamiltonian cycle if and only if $t^\ast = 2$.
	This completes the proof.
\end{proof}

Note that the reduction does not produce a planar graph
even when the input digraph has a planar underlying graph.
The example in \figurename~\ref{fig:Reduction_bipartite} contains
a $K_5$-minor.

The proofs actually show that \prob is \NP-hard to approximate within
a factor of less than $3/2$.


\section{Conclusion}
In this paper,
we studied the shortest reconfiguration problem of perfect matchings 
under the alternating cycle model, which is equivalent to
the combinatorial shortest path problem on perfect matching polytopes.
We prove that the problem can be solved in polynomial time for
outerplanar graphs, but it is {\NP}-hard, and even {\APX}-hard
for planar graphs and bipartite graphs.

Several questions remain unsolved.
For polynomial-time solvability, our algorithm runs only for
outerplanar graphs, and it looks difficult to extend the algorithm
to other graph classes.
A next step would be to try $k$-outerplanar graphs for fixed $k\geq 2$.

One way to tackle \NP-hard cases is approximation.
We only know the \NP-hardness of $3/2$-approximation.
We believe the existence of a polynomial-time constant-factor approximation.
Note that 
we do not obtain a constant-factor approximation
by flipping alternating cycles in the symmetric difference of 
two given perfect matchings one by one.

This paper was mainly concerned with reconfiguration of perfect matchings.
Alternatively, we may consider reconfiguration of maximum matchings, or
maximum-weight matchings.
In those cases, we need to adopt the alternating path/cycle model.
Then, the question is related to the combinatorial shortest path problem
on faces of matching polytopes.
Note that the perfect matching polytope is also a face of the matching 
polytope.
Therefore, the study on maximum-weight matchings will be a generalization
of this paper.

To the best of the authors' knowledge, the combinatorial shortest path
problem of $0/1$-polytopes has not been well investigated while
the adjacency in $0/1$-polytopes has been extensively studied in the
literature.
This paper opens up a new perspective for the study of combinatorial
and computational aspects of polytopes, and connects them with the
study of combinatorial reconfiguration.



\providecommand{\noopsort}[1]{}


\end{document}